\documentclass[a4paper,UKenglish]{lipics}
 
\usepackage{microtype}

\newtheorem{rem}[theorem]{\textbf{Remark}}

\def\hil{{\mathcal H}}

\def\B{{\mathcal B}}

\def\M{{\mathcal M}}
\def\N{{\mathcal N}}

\def\S{{\mathcal S}}

\def\X{{\mathcal X}}
\def\Y{{\mathcal Y}}
\def\bz{\left(}
\def\jz{\right)}

\def\ki{\textit}
\def\kii{\textsl}
\def\kiii{}

\def\half{\frac{1}{2}}
\def\ep{\varepsilon}

\def\bN{\mathbb{N}}

\def\bR{\mathbb{R}}

\def\prod{\mathrm{prod}}

\def\old{^{\mathrm{(old)}}}
\def\nw{^{\mathrm{(new)}}}
\def\inv{^{-1}}

\newcommand{\norm}[1]{\left\| #1\right\|}

\newcommand{\rsr}[3]{D_{#3}\bz #1\|#2\jz}
\newcommand{\rsro}[3]{D^{\mathrm{(old)}}_{#3}\bz #1\|#2\jz}
\newcommand{\rsrn}[3]{D^{\mathrm{(new)}}_{#3}\bz #1\|#2\jz}

\newcommand{\ds}{\mbox{ }\mbox{ }}

\DeclareMathOperator{\Tr}{Tr}

\DeclareMathOperator{\supp}{supp}

\title{Convexity properties of the quantum R\'enyi divergences, with applications to the quantum Stein's lemma 
\footnote{This work was partially supported by the European Research Council Advanced Grant ``IRQUAT''.}}
\titlerunning{} 

\author[1,2]{Mil\'an Mosonyi}
\affil[1]{F\'{\i}sica Te\`{o}rica: Informaci\'{o} i Fenomens Qu\`{a}ntics,
Universitat Aut\`{o}noma de Barcelona, ES-08193 Bellaterra (Barcelona), Spain}
\affil[2]{Mathematical Institute, Budapest University of Technology and Economics, \\
Egry J\'ozsef u~1., Budapest, 1111 Hungary.}
\authorrunning{M.~Mosonyi} 

\Copyright{Mil\'an Mosonyi}


\serieslogo{}
\volumeinfo
  {}
  {2}
  {}
  {1}
  {1}
  {1}
\EventShortName{}
\DOI{10.4230/LIPIcs.xxx.yyy.p}

\begin{document}

\maketitle

\begin{abstract}
We show finite-size bounds on the deviation of the optimal type II error from its asymptotic value in the quantum hypothesis testing problem of Stein's lemma with composite null-hypothesis. 
The proof is based on some simple properties of a new 
notion of quantum R\'enyi divergence, recently introduced in 
[M\"uller-Lennert, Dupuis, Szehr, Fehr and Tomamichel, 
J.~Math.~Phys. \textbf{54}, 122203, (2013)], and [Wilde, Winter, Yang, arXiv:1306.1586].
\end{abstract}

\section{Introduction}

R\'enyi defined the $\alpha$-divergence \cite{Renyi} of two probability distributions $p,q$ on a finite set $\X$ as
\begin{align*}
D_{\alpha}(p\|q):=\frac{1}{\alpha-1}\log\sum_{x\in\X}p(x)^{\alpha}q(x)^{1-\alpha},
\end{align*}
where $\alpha\in(0,+\infty)\setminus\{1\}$. 
These divergences have various desirable mathematical properties; they are strictly positive, non-increasing under 
stochastic maps, and 
jointly convex for $\alpha\in(0,1)$ and jointly quasi-convex for $\alpha>1$. 
For fixed $p$ and $q$,
$D_{\alpha}(p\|q)$ is a monotone increasing function of $\alpha$, and the 
limit $\alpha\to 1$ yields the relative entropy (a.k.a.~Kullback-Leibler divergence), probably the
single most important quantity in information theory.
Even more importantly, the R\'enyi divergences have great operational significance, as quantifiers of the trade-off between the 
relevant operational quantities in many information theoretic tasks, including hypothesis testing, source compression, and 
information transmission through noisy channels \cite{Csiszar}.
A direct operational interpretation of 
the R\'enyi divergences as generalized cutoff rates has been shown in \cite{Csiszar}.

In the view of the above, it is natural to look for an extension of the R\'enyi divergences for pairs of quantum states. 
One such extension has been known in quantum information theory for quite some time, defined for states $\rho$ and $\sigma$ as
\cite{OP}
\begin{align*}
D_{\alpha}\old(\rho\|\sigma):=\frac{1}{\alpha-1}\log\Tr\rho^{\alpha}\sigma^{1-\alpha}.
\end{align*}
These divergences also form a monotone increasing family, with the Umegaki relative entropy $D_1(\rho\|\sigma):=\Tr\rho(\log\rho-\log\sigma)$ as their limit at $\alpha\to 1$. They are also strictly positive; however, monotonicity under stochastic (i.e., completely positive and trace-preserving) maps only holds for $\alpha\in[0,2]$.
Recently, a new quantum R\'enyi divergence has been introduced in \cite{Renyi_new,WWY}, defined as
\begin{align*}
D_{\alpha}\nw(\rho\|\sigma):=\frac{1}{\alpha-1}\log\Tr\bz \sigma^{\frac{1-\alpha}{2\alpha}}\rho\sigma^{\frac{1-\alpha}{2\alpha}}\jz^{\alpha}.
\end{align*}
Again, these new divergences yield the Umegaki relative entropy in the limit $\alpha\to 1$, and monotonicity 
only holds on a restricted domain, in this case for $\alpha\in[1/2,+\infty)$.

Operational interpretation has been found for both definitions in the setting of binary hypothesis testing for different and 
matching domains of $\alpha$. The goal in hypothesis testing is to decide between two candidates, $\rho$ and $\sigma$, for the true state of a quantum system, based on a measurement 
on many identical copies of the system. The quantum Stein's lemma \cite{HP,ON} states that it is possible to make the probability of erroneously choosing $\rho$ (type II error)
to vanish exponentially fast in the number of copies, with the exponent being the relative entropy $D_1(\rho\|\sigma)$, while the probability of erroneously choosing $\sigma$
(type I error) goes to zero asymptotically. If the type II error is required to vanish with a suboptimal exponent $r<D_1(\rho\|\sigma)$ (this is called the direct domain) then the type I 
error can also be made to vanish 
exponentially fast, with the optimal exponent being the Hoeffding divergence $H_r:=\sup_{\alpha\in(0,1)}\frac{\alpha-1}{\alpha}[r-D_{\alpha}\old(\rho\|\sigma)]$
\cite{ANSzV,Hayashi,Nagaoka}. Thus, the $D_{\alpha}\old$ with $\alpha\in(0,1)$ quantify the trade-off between the 
rates of the type I and the type II error probabilities in the direct domain. 
Based on this trade-off relation, a 
more direct operational interpretation was obtained in \cite{MH} as generalized cutoff rates in the sense of Csisz\'ar 
\cite{Csiszar}. On the other hand, if the type II error is required to vanish with an exponent $r>D_1(\rho\|\sigma)$ (this is called the strong converse domain) then the type I 
error goes to $1$ 
exponentially fast, with the optimal exponent being the converse Hoeffding divergence $H_r^*:=\sup_{\alpha>1}\frac{\alpha-1}{\alpha}[r-D_{\alpha}\nw(\rho\|\sigma)]$
\cite{MO}. Thus, the $D_{\alpha}\nw$ with $\alpha>1$ quantify the 
trade-off between the rates of the type I success probability and the type II error probability in the strong converse region. Based on this, 
a direct operational interpretation of the $D_{\alpha}\nw$ as generalized cutoff rates was also given in \cite{MO}
for $\alpha>1$.

In the view of the above results, it seems that the old and the new definitions provide the operationally relevant quantum extension of R\'enyi's divergences in different domains: 
for $\alpha\in(0,1)$, the operationally relevant definition seems to be the old one, corresponding to the direct domain of hypothesis testing, whereas for 
$\alpha>1$, the operationally relevant definition seems to be the new one, corresponding to the strong converse domain of hypothesis testing. 

This is the picture at least when one wants to describe the full trade-off curve; most of the time, however, one is interested in one single point of this curve, corresponding to 
$\alpha=1$, where the transition from exponentially vanishing error probability to exponentially vanishing success probability happens.
It is known that using the ``wrong'' divergence can be beneficial to obtaining coding theorems at this point. Indeed, the strong converse property for hypothesis testing and 
classical-quantum channel coding has been proved using $D_{\alpha}\old$ for $\alpha>1$ in \cite{Nagaoka2,ON,ON2}
(``wrong'' divergence with the ``right'' values of $\alpha$), while a proof for the direct part of these problems was obtained recently 
in \cite{BG}, using $D_2\nw$ (```wrong'' divergence with a ``wrong'' value of $\alpha$).

Further examples of coding theorems based on the ``wrong'' R\'enyi divergence were given in \cite{Mosonyi}, where
it was shown that a certain concavity property of the new R\'enyi divergences, which the old ones don't have, make them a very convenient tool to prove
the direct part of various
coding theorems in composite/compound settings. This was demonstrated by giving short and simple proofs for the direct part of Stein's lemma with composite null-hypothesis
and for classical-quantum channel coding with compound channels.
Although the optimal rates for these problems have already been known \cite{BDKSSSz,BB,DD,Notzel}, the proofs in \cite{Mosonyi} are different from the previous ones, and offer considerable simplifications. 
The general approach is the following: 
\begin{enumerate}
\item
We start with a single-shot coding theorem that gives a trade-off relation between the relevant quantities of the problem in terms of R\'enyi divergences. For Stein's lemma, this is 
Audenaert's trace inequality \cite{Aud}, while for channel coding we use the Hayashi-Nagaoka random coding theorem from \cite{HN}.

\item
We then use general properties of the R\'enyi divergences to decouple the upper bounds from multiple to a single null-hypothesis/channel and to derive the asymptotics.
\end{enumerate}
The main advantage of this approach is that the second step only relies on universal properties of the R\'enyi divergences and is largely independent of the concrete problem at hand. 
In particular, the coding theorems for the composite/compound settings can be obtained with the same amount of effort as for a simple null-hypothesis/single channel.

In this paper we present a variant for the proof of Stein's lemma with composite null-hypothesis. While in \cite{Mosonyi} exponential bounds on the error probabilities were given, here we 
study the asymptotics of the optimal type II error probability for a given threshold $\ep$ on the type I error probability. Building on results from \cite{AMV} and \cite{Mosonyi}, we 
derive finite-size bounds on the deviation of the optimal type II error from its asymptotic value. Such bounds are of practical importance, since in real-life scenarios one always works 
with finitely many copies.

The structure of the paper is as follows. Section \ref{sec:not} is a summary of notations. In Section \ref{sec:Renyi} we review some properties of the quantum R\'enyi divergences,
including two inequalities from \cite{Mosonyi}:
Lemma \ref{lemma:old-new bounds}, which gives quantitative bounds between the old and the new definitions of the quantum R\'enyi divergences, and 
Corollary \ref{cor:new renyi superadd}, which shows that the convexity 
of the new R\'enyi divergence in its first argument can be complemented in the form of a weak quasi-concavity inquality. 
For readers' convenience, we include the proof of these inequalities.
In Section \ref{section:Stein} we prove the above mentioned finite-size version of Stein's lemma.

\section{Notations}
\label{sec:not}
   
For a finite-dimensional Hilbert space $\hil$, let $\B(\hil)_+$ denote the set of all 
non-zero positive semidefinite operators on $\hil$, and let 
$\S(\hil):=\{\rho\in\B(\hil)_+\,;\,\Tr\rho=1\}$ be the set of all density operators (states) 
on $\hil$.

We define the powers of a positive semidefinite operator $A$ only on its support; that is, 
if $\lambda_1,\ldots,\lambda_r$ are the strictly positive eigenvalues of $A$, with corresponding 
spectral projections $P_1,\ldots,P_r$, then we define
$A^{\alpha}:=\sum_{i=1}^r \lambda_i^{\alpha}P_i$ for all $\alpha\in\bR$. In particular, 
$A^0=\sum_{i=1}^rP_i$ is the projection onto the support of $A$, and we use $A^0\le B^0$ as a shorthand for $\supp A\subseteq\supp B$.

By a \ki{POVM (positive operator-valued measure)} $T$ on a Hilbert space $\hil$ we mean a map $T:\,\Y\to\B(\hil)$, where $\Y$ is some finite set,
$T(y)\ge 0$ for all $y$, 
and $\sum_{y\in\Y}T(y)=I$. In particular, a binary POVM is a POVM with $\Y=\{0,1\}$.

We denote the natural logarithm by $\log$, and 
use the convention $\log 0:=-\infty$ and $\log +\infty:=+\infty$.

\section{R\'enyi divergences}
\label{sec:Renyi}

For non-zero positive semidefinite operators $\rho,\sigma$, 
the \ki{R\'enyi $\alpha$-divergence} of $\rho$ w.r.t.~$\sigma$ 
with parameter $\alpha\in(0,+\infty)\setminus\{1\}$ is traditionally defined as \cite{OP}
\begin{align*}
\rsro{\rho}{\sigma}{\alpha}&:=
\begin{cases}
\frac{1}{\alpha-1}\log\Tr\rho^{\alpha}\sigma^{1-\alpha}-\frac{1}{\alpha-1}\log\Tr\rho,
& \alpha\in(0,1)\ds\text{or}\ds\rho^0\le\sigma^0,\\
+\infty,&\text{otherwise}.
\end{cases}
\end{align*}
For the mathematical properties of $D_{\alpha}\old$, see, e.g.~\cite{Lieb,MH,Petz}.
Recently, a new notion of R\'enyi divergence has been introduced in \cite{Renyi_new,WWY}, defined as
\begin{align*}
\rsrn{\rho}{\sigma}{\alpha}&:=
\begin{cases}
\frac{1}{\alpha-1}\log\Tr\bz\sigma^{\frac{1-\alpha}{2\alpha}}\rho\sigma^{\frac{1-\alpha}{2\alpha}}\jz^{\alpha}
-\frac{1}{\alpha-1}\log\Tr\rho,
& \alpha\in(0,1)\ds\text{or}\ds\rho^0\le\sigma^0,\\
+\infty,&\text{otherwise}.
\end{cases}
\end{align*}
For the mathematical properties of $D_{\alpha}\nw$, see, e.g.~\cite{Beigi,FL,MO,Renyi_new,WWY}.

An easy calculation shows that for fixed $\rho$ and $\sigma$, the function $\alpha\mapsto\log\Tr\rho^{\alpha}\sigma^{1-\alpha}$ is convex, which in turn yields immediately that 
$\alpha\mapsto \rsro{\rho}{\sigma}{\alpha}$ is monotone increasing. Moreover, the limit at $\alpha=1$ can be easily calculated as
\begin{align}\label{limit at 1}
D_1(\rho\|\sigma):=\lim_{\alpha\to 1}\rsro{\rho}{\sigma}{\alpha}
=
\begin{cases}
\frac{1}{\Tr\rho}\Tr\rho(\log\rho-\log\sigma),&\rho^0\le\sigma^0,\\
+\infty,&\text{otherwise},
\end{cases}
\end{align}
where the latter expression is \ki{Umegaki's relative entropy} \cite{Umegaki}. The same limit relation for $\rsrn{\rho}{\sigma}{\alpha}$ has been shown in 
\cite[Theorem 5]{Renyi_new}.
The following Lemma, due to 
\cite{TCR} and \cite{Tomamichel}, complements the above monotonicity property
around $\alpha=1$, and  in the same time
gives a quantitative version of \eqref{limit at 1}:

\begin{lemma}\label{prop:TCR}
Let $\rho,\sigma\in\B(\hil)_+$ be such that $\rho^0\le \sigma^0$, let
$\kappa:=\log(1+\Tr\rho^{3/2}\sigma^{-1/2}+\Tr\rho^{1/2}\sigma^{1/2})$, let $c>0$, and  
$\delta:=\min\left\{\half, \frac{c}{2\kappa}\right\}$. Then
\begin{align*}
\rsr{\rho}{\sigma}{1}&\ge\rsro{\rho}{\sigma}{\alpha}
\ge
\rsr{\rho}{\sigma}{1}-4(1-\alpha)\kappa^2\cosh c,\ds\ds\ds 1-\delta<\alpha<1,
\end{align*}
and the inequalities hold in the converse direction for $1<\alpha<1+\delta$.
\end{lemma}

\begin{rem}\label{rem:kappa}
Assume that $\rho$ and $\sigma$ are states.
The function $f(\alpha):=\Tr\rho^{\alpha}\sigma^{1-\alpha}$ is convex in $\alpha$, and $\rho^0\le\sigma^0$ implies that 
$f(1)=1$. Hence, $\alpha\mapsto (f(\alpha)-1)/(\alpha-1)$ is monotone increasing. Comparing the values at $1/2$ and $3/2$, we see 
that $\Tr\rho^{3/2}\sigma^{-1/2}+\Tr\rho^{1/2}\sigma^{1/2}\ge 2$, and thus $\kappa>1$.
\end{rem}

\begin{rem}
The \ki{R\'enyi entropy} of a positive semidefinite operator $\rho\in\B(\hil)_+$ with parameter $\alpha\in(0,+\infty)$ is defined as
\begin{align*}
S_{\alpha}(\rho):=-D_{\alpha}\old(\rho\|I)=-D_{\alpha}\nw(\rho\|I)=\frac{1}{1-\alpha}\log\Tr\rho^{\alpha}-\frac{1}{1-\alpha}\log\Tr\rho.
\end{align*}
By the above considerations, $\alpha\mapsto S_{\alpha}(\rho)$ is monotone decreasing, and comparing its values at 
$\alpha$ and at $0$, we get
\begin{align}\label{power bound}
\Tr\rho^{\alpha}\le(\Tr\rho^0)^{(1-\alpha)}(\Tr\rho)^{\alpha},\ds\ds\ds\alpha\in(0,1).
\end{align}
\end{rem}
\smallskip

According to the Araki-Lieb-Thirring inequality \cite{Araki,LT}, for any positive semidefinite operators 
$A,B$, $\Tr A^{\alpha}B^{\alpha}A^{\alpha}\le\Tr (ABA)^{\alpha}$
for $\alpha\in(0,1)$, and the inequality holds in the converse direction for
$\alpha>1$. A converse to the Araki-Lieb-Thirring inequality was  given in 
\cite{Aud-ALT}, where it was shown that 
$\Tr (ABA)^{\alpha}\le \bz\norm{B}^{\alpha}\Tr A^{2\alpha}\jz^{1-\alpha}
\bz\Tr A^{\alpha}B^{\alpha}A^{\alpha}\jz^{\alpha}$
for $\alpha\in(0,1)$, and the inequality holds in the converse direction for $\alpha>1$.
Applying these inequalities to $A:=\rho^{\half}$ and 
$B:=\sigma^{\frac{1-\alpha}{\alpha}}$, we get
\begin{align}\label{old-new bounds0}
\Tr\rho^{\alpha}\sigma^{1-\alpha}
\le
\Tr\bz\rho^{\half}\sigma^{\frac{1-\alpha}{\alpha}}\rho^{\half}\jz^{\alpha}
\le
\norm{\sigma}^{(1-\alpha)^2}\bz\Tr\rho^{\alpha}\jz^{1-\alpha}\bz\Tr\rho^{\alpha}\sigma^{1-\alpha}\jz^{\alpha}
\end{align}
for $\alpha\in(0,1)$, and the inequalities hold in the converse direction for $\alpha>1$.
In terms of the R\'enyi divergences, the above inequalities yield the ones in the following Lemma, the first of which 
has already been pointed out in \cite{WWY} and \cite{DL}.

\begin{lemma}\label{lemma:old-new bounds}
Let $\rho,\sigma\in\S(\hil)$ be states. For any $\alpha\in(0,+\infty)$,
\begin{align}\label{old-new bounds}
\rsro{\rho}{\sigma}{\alpha}
\ge
\rsrn{\rho}{\sigma}{\alpha}
\ge
\alpha\rsro{\rho}{\sigma}{\alpha}-|\alpha-1|\log\dim\hil.
\end{align}
\end{lemma}
\begin{proof}
The first inequality is immediate from the first inequality in \eqref{old-new bounds0}. Taking into account \eqref{power bound}, and that $\norm{\sigma}\le 1$, the second inequality in 
\eqref{old-new bounds0} yields the second inequality in \eqref{old-new bounds} for $\alpha\in(0,1)$. For $\alpha>1$, we have 
$\Tr(\rho/\norm{\rho})^{\alpha}\le\Tr(\rho/\norm{\rho}$, and hence we get 
$\Tr\bz\rho^{\half}\sigma^{\frac{1-\alpha}{\alpha}}\rho^{\half}\jz^{\alpha}
\ge
\norm{\sigma}^{(1-\alpha)^2}\norm{\rho}^{-(\alpha-1)^2}\bz\Tr\rho^{\alpha}\sigma^{1-\alpha}\jz^{\alpha}$.
Using that $\norm{\rho}\le 1$ and that $\norm{\sigma}\ge 1/\dim\hil$, we get the second inequality in 
 \eqref{old-new bounds} for $\alpha>1$.
\end{proof}
\medskip

For $\rho,\sigma\in\B(\hil)_+$, let
\begin{equation}\label{new Q}
Q\old_{\alpha}(\rho\|\sigma):=\Tr\rho^{\alpha}\sigma^{1-\alpha},\ds\ds\ds\ds\ds\ds
Q\nw_{\alpha}(\rho\|\sigma):=
\Tr\bz\sigma^{\frac{1-\alpha}{2\alpha}}\rho\sigma^{\frac{1-\alpha}{2\alpha}}\jz^{\alpha}
\end{equation}
be the core quantities of the R\'enyi divergences $D_{\alpha}\old$ and $D_{\alpha}\nw$, respectively.
$Q_{\alpha}\old$ is jointly concave in $(\rho,\sigma)$ for $\alpha\in[0,1]$ (see \cite{Lieb,Petz}) and jointly convex for $\alpha\in[1,2]$ (see \cite{Ando,Petz}).
The general concavity result in \cite[Theorem 2.1]{Hiai} implies as a special case that 
$Q\nw_{\alpha}(\rho\|\sigma)$ is jointly concave in $(\rho,\sigma)$ for 
$\alpha\in[1/2,1)$. (See also \cite{FL} for a different proof of this). 
In \cite{Renyi_new,WWY}, joint convexity of $Q\nw_{\alpha}$ was shown for $\alpha\in[1,2]$, which was later extended in \cite{FL}, using a different proof method, to all $\alpha>1$.
These results are equivalent to the monotonicity of the R\'enyi divergences under completely positive trace-preserving maps, for 
$\alpha\in[0,2]$ in the case of $D_{\alpha}\old$, and for $\alpha\ge 1/2$ in the case of $D_{\alpha}\nw$.

The next lemma shows that the concavity of $Q_{\alpha}\nw$ in its first argument can be complemented by a subadditivity inequality for $\alpha\in(0,1)$:
\begin{lemma}\label{lemma:concavity complement}
Let $\rho_1,\ldots,\rho_r\in\S(\hil)$ be states and $\sigma\in\B(\hil)_+$, and let 
$\gamma_1,\ldots,\gamma_r$ be a probability distribution. For every $\alpha\in(0,1)$, 
\begin{align}
\sum_i\gamma_iQ\nw_{\alpha}(\rho_i\|\sigma)&\le Q\nw_{\alpha}\bigg(\sum_i\gamma_i\rho_i\Big\|\sigma\bigg)\le \sum_i\gamma_i^{\alpha}Q\nw_{\alpha}(\rho_i\|\sigma).\label{concavity complement}
\end{align}
\end{lemma}
\begin{proof}
The function $x\mapsto x^{\alpha}$ is operator concave on $[0,+\infty)$ for $\alpha\in(0,1)$ (see Theorems V.1.9 and V.2.5 in  \cite{Bhatia}), from which the first inequality in \eqref{concavity complement} follows immediately. To prove the second inequality, we use a special case of the Rotfel'd inequality, for which we provide a proof below.
First let $A,B\in\B(\hil)_+$ be invertible. Then 
\begin{align}
\Tr(A+B)^{\alpha}-\Tr A^{\alpha}&=
\int_{0}^1\frac{d}{dt}\Tr(A+tB)^{\alpha}\,dt
=
\int_{0}^1\alpha\Tr B(A+tB)^{\alpha-1}\,dt\nonumber\\
&\le
\int_{0}^1\alpha\Tr B(tB)^{\alpha-1}\,dt
=\Tr B^{\alpha}\int_{0}^1\alpha t^{\alpha-1}\,dt
=\Tr B^{\alpha},\label{subadditivity proof}
\end{align}
where in the first line we used
the identity $(d/dt)\Tr f(A+tB)=\Tr Bf'(A+tB)$, and the inequality follows
from the fact that $x\mapsto x^{\alpha-1}$ is operator monotone decreasing on $(0,+\infty)$ for 
$\alpha\in(0,1)$. By continuity, we can drop the invertibility assumption, and
\eqref{subadditivity proof} holds for any $A,B\in\B(\hil)_+$. Obviously, \eqref{subadditivity proof} extends to more than two operators, i.e., $\Tr(A_1+\ldots+A_r)^{\alpha}\le \Tr A_1^{\alpha}+\ldots+\Tr A_r^{\alpha}$ for any 
$A_,\ldots, A_r\in\B(\hil)_+$ and $\alpha\in(0,1)$. Choosing now $A_i:=\sigma^{\frac{1-\alpha}{2\alpha}}\gamma_i\rho_i\sigma^{\frac{1-\alpha}
{2\alpha}}$ yields the second inequality in \eqref{concavity complement}.
\end{proof}

\begin{corollary}\label{cor:new renyi superadd}
Let $\rho_1,\ldots,\rho_r\in\S(\hil)$ be states  and $\sigma\in\B(\hil)_+$, and let $\gamma_1,\ldots,\gamma_r$ be a probability 
distribution. 
For every $\alpha\in(0,1)$,
\begin{align*}
\min_i\rsrn{\rho_i}{\sigma}{\alpha}+\log \min_i\gamma_i
\le 
D\nw_{\alpha}\bigg(\sum_i\gamma_i\rho_i\Big\|\sigma\bigg)
\le\sum_i\gamma_i\rsrn{\rho_i}{\sigma}{\alpha}.
\end{align*} 
\end{corollary}
\begin{proof}
Immediate from Lemma \ref{lemma:concavity complement}.
\end{proof}

\section{Stein's lemma with composite null-hypothesis}
\label{section:Stein}

In the general formulation of binary quantum hypothesis testing, we assume that for every $n\in\bN$, 
a quantum system with Hilbert space $\hil_n$ is given, together with
two subsets 
$H_{0,n}$ and $H_{1,n}$ of the state space of $\hil_n$, corresponding to the \ki{null-hypothesis} and 
the \ki{alternative hypothesis}, respectively. Our aim is to guess, based on a binary POVM, which set 
the true state of the system falls into. Here we consider the i.i.d.~case with composite null-hypothesis and simple alternative 
hypothesis. That is, for every $n\in\bN$, $\hil_n=\hil^{\otimes n}$ for some finite-dimensional Hilbert space $\hil$; the 
null-hypothesis is represented by a set of states
$\N\subseteq\S(\hil)$, and the alternative hypothesis is represented by a single state $\sigma\in\S(\hil)$.
For every $n\in\bN$, we have $H_{0,n}=\N^{(\otimes n)}:=\{\rho^{\otimes n}:\,\rho\in\N\}$, and 
$H_{1,n}=\{\sigma^{\otimes n}\}$. 

Given a binary POVM $T_n=(T_n(0),T_n(1))$, with $T_n(0)$ corresponding to accepting the null-hpothesis and $T_n(1)$ to accepting 
the alternative hypothesis, there are two possible ways of making an erroneous decision: accepting the alternative hypothesis when the null-hypothesis is true, called the type I error, or the other way around, called the type II error. The probabilities of these two errors are given by
\begin{align*}
\alpha_n(T_n):=\sup_{\rho\in\N}\Tr\rho^{\otimes n}T_n(1),\ds\text{(type I)\ds\ds\ds and\ds\ds\ds}
\beta_n(T_n):=\Tr\sigma^{\otimes n}T_n(0),\ds\text{(type II)}.
\end{align*}
Note that in the definition of $\alpha_n$, we used a worst-case error probability.

In the setting of Stein's lemma, one's aim is to keep the type I error below a threshold $\ep$, and to optimize the type II error under this condition. For any set $\M\subseteq\S(\hil^{\otimes n})$ and any $\ep\in(0,1)$, let 
\begin{align*}
\beta_{\ep}(\M\|\sigma^{\otimes n}):=\inf\left\{\Tr\sigma^{\otimes n}T_n(0):\,\sup_{\omega\in\M}\Tr\omega T_n(1)\le\ep\right\},
\end{align*} 
where the infimum is taken over all binary POVM $T_n$ on $\hil^{\otimes n}$. When $\M$ consists of one single element $\omega$, we simply write $\beta_{\ep}(\omega\|\sigma^{\otimes n})$. The quantum Stein's lemma states that 
\begin{align}\label{Stein's lemma}
\lim_{n\to+\infty}\frac{1}{n}\log\beta_{\ep}\bz \N^{(\otimes n)}\|\sigma^{\otimes n}\jz=
-D_1(\N\|\sigma):=-\inf_{\rho\in\N}D_1(\rho\|\sigma).
\end{align}
This has been shown first in \cite{HP,ON2} for the case where $\N$ consists of one single element $\rho$.
Theorem 2 in \cite{Hayashi_Stein} uses group representation techniques to give an approximation of the relative entropy in terms 
of post-measurement relative entropies, which, when combined with Stein's lemma for probability distributions, yields 
\eqref{Stein's lemma} for finite $\N$. A direct proof for the case of infinite $\N$, also based on group representation theory, 
has recently been given in \cite{Notzel}. A version of Stein's lemma with infinite $\N$ has been previously 
proved in \cite{BDKSSSz}, however, with a weaker error criterion.

Here we give a different proof of the quantum Stein's lemma with possibly infinite composite null-hypothesis. Our proof is based 
on the results of \cite{AMV}, where bounds on $\beta_\ep$ were obtained in terms of R\'enyi divergences, and general properties of the R\'enyi divergences from Section \ref{sec:Renyi}. Moreover, we give a refined version of 
\eqref{Stein's lemma} in Theorem \ref{thm:Stein} by providing finite-size corrections to the deviation of 
$\frac{1}{n}\log\beta_{\ep}\bz \N^{(\otimes n)}\|\sigma^{\otimes n}\jz$ from its asymptotic value $-D_1(\N\|\sigma)$
for every $n\in\bN$.

We will need the following results from \cite{AMV}:

\begin{lemma}\label{lemma:AMV}
Let $\rho,\sigma\in\S(\hil)$. For every $\ep\in(0,1)$ and every $\alpha\in(0,1)$,
\begin{align}\label{beta upper3}
\log\beta_{\ep}(\rho\|\sigma)\le -D_{\alpha}\old(\rho\|\sigma)+\frac{\alpha}{1-\alpha}\log\ep\inv-\frac{h_2(\alpha)}{1-\alpha},
\end{align}
where $h_2(\alpha):=-\alpha\log\alpha-(1-\alpha)\log(1-\alpha)$ is the binary entropy function. Moreover, for every $n\in\bN$,
\begin{align}\label{beta lower2}
\frac{1}{n}\log\beta_{\ep}\bz\rho^{\otimes n}\|\sigma^{\otimes n}\jz\ge -D_1(\rho\|\sigma)-\frac{1}{\sqrt{n}}4\sqrt{2}\kappa\log(1-\ep)\inv,
\end{align}
where $\kappa$ is given in Lemma \ref{prop:TCR}.
\end{lemma}
\begin{proof}
The upper bound \eqref{beta upper3} is due to \cite[Proposition 3.2]{AMV}, while the lower bound in \eqref{beta lower2} is 
formula (19) in \cite[Theorem 3.3]{AMV}.
\end{proof}

When $\N$ is infinite, we will need the following approximation lemma, which is a special case of 
\cite[Lemma 2.6]{MS}:
\begin{lemma}\label{lemma:approximation}
For every $\delta>0$, let $\N_{\delta}\subset\N$ be a set of minimal cardinality such that 
$\sup_{\rho\in\N}\inf_{\rho'\in\N_{\delta}}\norm{\rho-\rho'}_1\le \delta$. 
Then $|\N_{\delta}|\le \min\{|\N|,(1+2\delta\inv)^D\}$, where $D=(\dim\hil+1)(\dim\hil)/2$, and
\begin{align}\label{Stein approximation}
\sup_{\rho\in\N}\inf_{\rho'\in\N_{\delta}}\norm{\rho^{\otimes n}-(\rho')^{\otimes n}}_1
\le 
n
\sup_{\rho\in\N}\inf_{\rho'\in\N_{\delta}}\norm{\rho-\rho'}_1
\le
n\delta,\ds\ds\ds n\in\bN.
\end{align}
\end{lemma}
\smallskip

Now we are ready to prove our main result:
\begin{theorem}\label{thm:Stein}
Let $\ep\in(0,1)$, and for 
every $n\in\bN$, let $0\le\delta_n\le\ep/(2n)$.
Then
\begin{align}
\frac{1}{n}\log\beta_{\ep}\bz \N^{(\otimes n)}\|\sigma^{\otimes n}\jz\le
&- D_1(\N\|\sigma)\nonumber\\
&+\sqrt{\frac{\log\bz 2|\N_{\delta_n}|\ep\inv\jz}{n}}\cdot
2\left[8\kappa_{\max}^2+\log\dim\hil+D_1(\N\|\sigma)\right]^{\half}\nonumber\\
&+\frac{\log\bz 2|\N_{\delta_n}|\ep\inv\jz}{n}\cdot 4\kappa_{\max},\label{beta upper2}\\
\frac{1}{n}\log\beta_{\ep}\bz \N^{(\otimes n)}\|\sigma^{\otimes n}\jz\ge
&- D_1(\N\|\sigma)-\frac{1}{\sqrt{n}}4\sqrt{2}\log(1-\ep)\inv\kappa_{\max},\label{beta lower}
\end{align}
where $\kappa_{\max}:=\sup_{\rho\in\N}\{\log(1+\Tr\rho^{3/2}\sigma^{-1/2}+\Tr\rho^{1/2}\sigma^{1/2})\}
\le \log(2+\Tr\sigma^{-1/2})<+\infty$.

In \eqref{beta upper2}, the slowest decaying term after $- D_1(\N\|\sigma)$ is of the order $1/\sqrt{n}$ when $\N$ is finite, and when $\N$ is infinite, it can be chosen to be of the order $\sqrt{\frac{\log n}{n}}$.
\end{theorem}
\begin{proof}
The lower bound in \eqref{beta lower} is immediate from \eqref{beta lower2}, and hence we only have to prove \eqref{beta upper2}.
We have
\begin{align*}
\log\beta_{\ep}\bz \N^{(\otimes n)}\|\sigma^{\otimes n}\jz
&\le
\log\beta_{\ep-n\delta_n}\bz \N_{\delta_n}^{(\otimes n)}\|\sigma^{\otimes n}\jz
\le
\log\beta_{\frac{\ep-n\delta_n}{|\N_{\delta_n}|}}\bz \sum_{\rho\in\N_{\delta_n}}\frac{1}{|\N_{\delta_n}|}\rho^{\otimes n}\bigg\|\sigma^{\otimes n}\jz\\
&\le
-D_{\alpha}\old\bz\sum_{\rho\in\N_{\delta_n}}\frac{1}{|\N_{\delta_n}|}\rho^{\otimes n}\bigg\|\sigma^{\otimes n}\jz
+\frac{\alpha}{1-\alpha}\log\frac{|\N_{\delta_n}|}{\ep-n\delta_n}\\
&\le
-D_{\alpha}\nw\bz\sum_{\rho\in\N_{\delta_n}}\frac{1}{|\N_{\delta_n}|}\rho^{\otimes n}\bigg\|\sigma^{\otimes n}\jz
+\frac{\alpha}{1-\alpha}\log\frac{|\N_{\delta_n}|}{\ep-n\delta_n},
\end{align*}
where the first inequality is due to \eqref{Stein approximation}, the second inequality is obvious, the third one follows from \eqref{beta upper3},
and the last one is due to Lemma \ref{lemma:old-new bounds}.
Note that $\ep-n\delta_n\ge \ep/2$ by assumption.
Using Corollary \ref{cor:new renyi superadd}, we can continue the above upper bound as
\begin{align*}
&\log\beta_{\ep}\bz \N^{(\otimes n)}\|\sigma^{\otimes n}\jz\\
&\ds\le
-\min_{\rho\in\N_{\delta_n}}D_{\alpha}\nw\bz\rho^{\otimes n}\|\sigma^{\otimes n}\jz+\log|\N_{\delta_n}|+
+\frac{\alpha}{1-\alpha}\log|\N_{\delta_n}|+\frac{\alpha}{1-\alpha}\log\frac{2}{\ep}\\
&\ds\le
-n\inf_{\rho\in\N}D_{\alpha}\nw\bz\rho\|\sigma\jz
+\frac{1}{1-\alpha}\log|\N_{\delta_n}|+\frac{1}{1-\alpha}\log\frac{2}{\ep},
\end{align*}
where in the last line we used the additivity property $D_{\alpha}\nw\bz\rho^{\otimes n}\|\sigma^{\otimes n}\jz=nD_{\alpha}\nw\bz\rho\|\sigma\jz$.

By Lemmas \ref{lemma:old-new bounds} and \ref{prop:TCR}, for every $\alpha\in(1/2,1)$ such that $\alpha>1-\frac{c}{2\kappa_{\max}}$,
\begin{align*}
\inf_{\rho\in\N}\rsrn{\rho}{\sigma}{\alpha}
&\ge
\alpha\inf_{\rho\in\N}\rsro{\rho}{\sigma}{\alpha}-(1-\alpha)\log\dim\hil\\
&\ge
\alpha\inf_{\rho\in\N}\rsr{\rho}{\sigma}{1}-4\alpha(1-\alpha)\kappa_{\max}^2\cosh c-(1-\alpha)\log\dim\hil,
\end{align*}
where $c$ is an arbitrary positive constant. Now choose $\alpha:=1-a/\sqrt{n}$. Then 
\begin{align*}
\frac{1}{n}\log\beta_{\ep}\bz \N^{(\otimes n)}\|\sigma^{\otimes n}\jz\le&
-\bz 1-\frac{a}{\sqrt{n}}\jz D_1(\N\|\sigma)+\frac{a}{\sqrt{n}}\bz 4\kappa_{\max}^2\cosh c+\log\dim\hil\jz\\
&+\frac{1}{a\sqrt{n}}\bz \log|\N_{\delta_n}|+\log\frac{2}{\ep}\jz.
\end{align*}
Optimizing over $a$ yields
\begin{align}
&\frac{1}{n}\log\beta_{\ep}\bz \N^{(\otimes n)}\|\sigma^{\otimes n}\jz\nonumber\\
&\ds\ds\le - D_1(\N\|\sigma)
+\frac{2}{\sqrt{n}}
\left[4\kappa_{\max}^2\cosh c+\log\dim\hil+D_1(\N\|\sigma)\right]^{\half}
\cdot\left[\log(2|\N_{\delta_n}|\ep\inv)\right]^{\half}.\label{beta upper}
\end{align}
The optimum is reached at 
\begin{align*}
a^*=\left[\log(2|\N_{\delta_n}|\ep\inv)\right]^{\half}\cdot 
\left[4\kappa_{\max}^2\cosh c+\log\dim\hil+D_1(\N\|\sigma)\right]^{-\half},
\end{align*}
and we need 
$a^*/\sqrt{n}\le 1/2$ and $a^*/\sqrt{n}\le c/(2\kappa_{\max})$, which is satisfied if 
\begin{align*}
\kappa_{\max}^2\cosh c\ge\frac{1}{n}\log(2|\N_{\delta_n}|\ep\inv)\ds\ds\ds\text{and}\ds\ds\ds
c^2\cosh c\ge\frac{1}{n}\log(2|\N_{\delta_n}|\ep\inv).
\end{align*}
Let us choose $c>0$ such that $\cosh c=2+\frac{1}{n}\log(2|\N_{\delta_n}|\ep\inv)$.
By Remark \ref{rem:kappa}, $\kappa_{\max}>1$, and hence the first inequality is satisfied. Moreover,
with this choice $c>1$, and thus the second inequality is satisfied as well.

Substituting this choice of $c$ into \eqref{beta upper}, and using the subadditivity of the square root, we get 
\eqref{beta upper2}.

When $\N$ is finite, we can choose $\delta_n=0$, and hence $\N_{\delta_n}=\N$, for every $n$. This shows that the second term in \eqref{beta upper2} is of the order 
$1/\sqrt{n}$, while the third term is of the order $1/n$. When $\N$ is infinite, we can choose $\delta_n=\ep/(2n^2)$, whence the order 
of the second term in \eqref{beta upper2} is $\sqrt{\frac{\log n}{n}}$, and the order of the third term is 
$\frac{\log n}{n}$.
\end{proof}

\begin{rem}
In the case of a simple null-hypothesis $\N=\{\rho\}$, the limit 
\begin{align}\label{second order}
\lim_{n\to+\infty}\sqrt{n}\bz\frac{1}{n}\log\beta_{\ep}(\N^{(\otimes n)}\|\sigma^{\otimes n})+D_1(\N\|\sigma)\jz,
\end{align} 
called the second-order asymptotics,
has been determined in 
\cite{KeLi,TH}. Their results show that the finite-size bounds of \cite{AMV} are not asymptotically optimal, and hence the same holds for the bounds in 
Theorem \ref{thm:Stein}. The merit of these latter results, on the other hand, is that the correction terms are easily computable, and the bounds are valid for any finite $n$. 
To the best of our knowledge, the value of the limit \eqref{second order} has not yet been determined when $|\N|>1$, and our bounds in 
Theorem \ref{thm:Stein} give bounds on the second-order asymptotics in this case.
\end{rem}

\subparagraph*{Acknowledgements}

The author is grateful to Professor Fumio Hiai and Nilanjana Datta for discussions.


\begin{thebibliography}{50}

\bibitem{Ando}
T.~Ando.
\kii{Concavity of certain maps and positive definite matrices and applications to
Hadamard products}.
\kiii{Linear Algebra Appl.} {\bf 26}, 203--241 1979
 
\bibitem{Araki}
H.~Araki.
\kii{On an inequality of Lieb and Thirring}.
\kiii{Letters in Mathematical Physics}, Volume 19, Issue 2, pp.~167--170, 1990

\bibitem{Aud} K.M.R.~Audenaert, J.~Calsamiglia, Ll.~Masanes, R.~Munoz-Tapia, A.~Acin, E.~Bagan, F.~Verstraete.
\kii{Discriminating states: the quantum Chernoff bound}.
\kiii{Phys.~Rev.~Lett.} \textbf{98} 160501, 2007

\bibitem{ANSzV} K.M.R.~Audenaert, M.~Nussbaum, A.~Szko\l a, F.~Verstraete.
\kii{Asymptotic error rates in quantum hypothesis testing}.
\kiii{Commun.~Math.~Phys.} {\bf 279}, 251--283, 2008

\bibitem{Aud-ALT}
K.M.R.~Audenaert.
\kii{On the Araki-Lieb-Thirring inequality}.
\kiii{Int.~J.~of Information and Systems Sciences} \textbf{4}, pp.~78--83, 2008)

\bibitem{AMV}
Koenraad M.R.~Audenaert, Milan Mosonyi, Frank Verstraete.
\kii{Quantum state discrimination bounds for finite sample size}.
\kiii{J.~Math.~Phys.} \textbf{53}, 122205, 2012

\bibitem{Beigi}
Salman Beigi.
\kii{Quantum R\'enyi divergence satisfies data processing inequality}.
\kiii{J.~Math.~Phys.}, \textbf{54}, 122202, 2013

\bibitem{BG}
Salman Beigi, Amin Gohari.
\kii{Quantum Achievability Proof via Collision Relative Entropy}.
\kiii{arXiv:1312.3822}, 2013

\bibitem{Bhatia}
R.~Bhatia.
\kii{Matrix Analysis}.
\kiii{Graduate Texts in Mathematics \textbf{169}, Springer}, 1997

\bibitem{BDKSSSz}
I.~Bjelakovic, J.-D.~Deuschel, T.~Kr\"uger, R.~Seiler, R.~Siegmund-Schultze, A.~Szko\l a.
\kii{A quantum version of Sanov's theorem}.
\kiii{Commun.~Math.~Phys.} \textbf{260}, pp.~659--671, 2005
    
\bibitem{BB}
I.~Bjelakovic, H.~Boche.
\kii{Classical capacities of compound and averaged quantum channels}.
\kiii{IEEE Trans.~Inform.~Theory} \textbf{55}, 3360--3374, 2009

\bibitem{Csiszar} I.~Csisz\'ar.
\kii{Generalized cutoff rates and R\'enyi's information measures}.
\kiii{IEEE Trans.~Inf.~Theory} \textbf{41}, 26--34, 1995

\bibitem{DD} 
N.~Datta, T.C.~Dorlas.
\kii{The Coding Theorem for a Class of Quantum Channels with Long-Term Memory}.
\kiii{Journal of Physics A: Mathematical and Theoretical}, vol.~40, 8147, 2007

\bibitem{DL}
Nilanjana Datta and Felix Leditzky.
\kii{A limit of the quantum Rényi divergence}.
\kiii{J.~Phys.~A: Math.~Theor.} \textbf{47} 045304, 2014

\bibitem{FL}
Rupert L.~Frank and Elliott H.~Lieb.
\kii{Monotonicity of a relative R\'enyi entropy}.
\kiii{J.~Math.~Phys.} 54 , 122201, 2013

\bibitem{Hayashi_Stein}
Masahito Hayashi.
\kii{Asymptotics of quantum relative entropy from a representation theoretical viewpoint}.
\kiii{J.~Phys.~A: Math.~Gen.} \textbf{34} 3413, (2001)

\bibitem{HN} M.~Hayashi, H.~Nagaoka.
\kii{General Formulas for Capacity of Classical-Quantum Channels}.
\kiii{IEEE Trans.~Inf.~Theory} \textbf{49}, 2003

\bibitem{Hayashi}
M.~Hayashi.
\kii{Error exponent in asymmetric quantum hypothesis testing and its
 application to classical-quantum channel coding}.
\kiii{Phys.~Rev.~A} \textbf{76}, 062301, 2007

\bibitem{HP}
F.~Hiai, D.~Petz.
\kii{The proper formula for relative entropy and its asymptotics in quantum probability}.
\kiii{Comm.~Math.~Phys.} {\bf 143}, 99--114, 1991

\bibitem{Hiai}
F.~Hiai.
\kii{Concavity of certain matrix trace and norm functions}.
\kiii{Linear Algebra and Appl.} \textbf{439}, 1568--1589,  2013

\bibitem{KeLi}
Ke Li.
\kii{Second-order asymptotics for quantum hypothesis testing}.
\kiii{Annals of Statistics}, Vol.~42, No.~1, pp.~171--189, 2014

\bibitem{Lieb}
E.H.~Lieb.
\kii{Convex trace functions and the Wigner-Yanase-Dyson conjecture}.
\kiii{Adv.~Math.} {\bf 11}, 267--288, 1973

\bibitem{LT}
E.H.~Lieb, W.~Thirring.
\kii{Studies in mathematical physics}.
pp.~269--297. Princeton
University Press, Princeton, 1976

\bibitem{MS}
Vitali D.~Milman, Gideon Schechtman.
\kii{Asymptotic Theory of Finite Dimensional Normed Spaces}.
\kiii{Lecture Notes in Mathematics}, Springer-Verlag Berlin Heidelberg, 1986

\bibitem{MH}
M.~Mosonyi, F.~Hiai.
\kii{On the quantum R\'enyi relative entropies and related capacity formulas}.
\kiii{IEEE Trans.~Inf.~Theory}, \textbf{57}, 2474--2487, 2011

\bibitem{MO}
M.~Mosonyi, T.~Ogawa.
\kii{Quantum hypothesis testing and the operational interpretation of the quantum R\'enyi relative entropies}.
\kiii{arXiv:1308.3228}, 2013

\bibitem{Mosonyi}
M.~Mosonyi.
\kii{Inequalities for the quantum R\'enyi divergences with applications to compound coding problems}.
\kiii{arXiv:1310.7525; submitted to IEEE Transactions on Information Theory}

\bibitem{Renyi_new}
M.~M\"uller-Lennert, F.~Dupuis, O.~Szehr, S.~Fehr, M.~Tomamichel.
\kii{On quantum Renyi entropies: a new definition and some properties}.
\kiii{J.~Math.~Phys.} \textbf{54}, 122203, 2013

\bibitem{Nagaoka2} H.~Nagaoka.
\kii{Strong converse theorems in quantum information theory}.
\kiii{in the book ``Asymptotic Theory of Quantum Statistical Inference'' edited by M.~Hayashi},
World Scientific, 2005

\bibitem{Nagaoka} H.~Nagaoka.
\kii{The converse part of the theorem for quantum Hoeffding bound}.
\kiii{quant-ph/0611289}, 2006

\bibitem{Notzel}
J.~N\"otzel.
\kii{Hypothesis testing on invariant subspaces of the symmetric group, part I - quantum Sanov's theorem and arbitrarily varying sources}.
\kiii{arXiv:1310.5553}, 2013

\bibitem{ON} 
T.~Ogawa, H.~Nagaoka.
\kii{Strong converse to the quantum channel coding theorem}.
\kiii{IEEE Transactions on Information Theory}, vol.~45, no.~7, pp.~2486-2489, 1999

\bibitem{ON2}
T.~Ogawa, H.~Nagaoka.
\kii{Strong converse and Stein's lemma in quantum hypothesis testing}.
\kiii{IEEE Trans. Inform. Theory} {\bf 47}, 2428--2433, 2000

\bibitem{OP}
M.~Ohya, D.~Petz.
\kii{Quantum Entropy and its Use}.
\kiii{Springer}, 1993

\bibitem{Petz} D.~Petz.
\kii{Quasi-entropies for finite quantum systems}.
\kiii{Rep.~Math.~Phys.} \textbf{23}, 57--65, 1986

\bibitem{Renyi}
A.~R\'enyi.
\kii{On measures of entropy and information}.
\kiii{Proc.~4th Berkeley Sympos.~Math.~Statist.~and Prob.}, Vol.~I, pp.~547--561, Univ. California Press, Berkeley, California,  1961

\bibitem{TCR}
M.~Tomamichel, R.~Colbeck, R.~Renner.
\kii{A fully quantum asymptotic equipartition property}.
\kiii{IEEE Trans.~Inform.~Theory} \textbf{55}, 5840--5847, 2009

\bibitem{Tomamichel}
M.~Tomamichel.
\kii{A framework for non-asymptotic quantum information theory}.
\kiii{PhD thesis, ETH Z\"urich}, 2012

\bibitem{TH}
M.~Tomamichel, M.~Hayashi.
\kii{A Hierarchy of Information Quantities for Finite Block Length Analysis of Quantum Tasks}.
\kiii{IEEE Transactions on Information Theory} \textbf{59}, pp.~7693--7710, 2013

\bibitem{Umegaki}
H.~Umegaki.
\kii{Conditional expectation in an operator algebra}.
\kiii{Kodai Math.~Sem.~Rep.} \textbf{14}, 59--85, 1962

\bibitem{WWY}
Mark M.~Wilde, Andreas Winter, Dong Yang.
\kii{Strong converse for the classical capacity of entanglement-breaking and Hadamard channels}.
\kiii{arXiv:1306.1586}, 2013

\end{thebibliography}
\end{document}